\documentclass{article}

\usepackage{arxiv}

\usepackage[utf8]{inputenc} 
\usepackage[T1]{fontenc}    
\usepackage{hyperref}       
\usepackage{url}            
\usepackage{booktabs}       
\usepackage{amsfonts}       
\usepackage{amsmath}
\usepackage{nicefrac}       
\usepackage{microtype}      
\usepackage{cleveref}       
\usepackage{graphicx}
\usepackage{natbib}
\usepackage{doi}
\usepackage{cite}
\usepackage{mathtools,amssymb}
\newtheorem{theorem}{Theorem}{}
\newtheorem{assumption}{Assumption}{}
\newtheorem{lemma}{Lemma}{}

{}
\newenvironment{proof}{{\noindent\it Proof.}\quad}{\hfill $\square$\par}
{}
\title{Distributed Containment Reference Signal for Nonholonomic Planar Vehicles}

\author{ \href{https://https://orcid.org/0000-0002-0174-4795}{\includegraphics[scale=0.06]{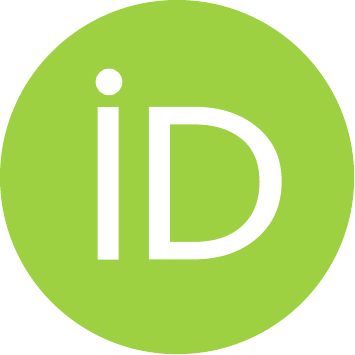}\hspace{1mm}Lixia Yan} \\
	The Seventh Research Division, School of Automation Science and Electrical Engineering\\
	Beihang University\\
	Beijing 100191, People's Republic of China\\
	\texttt{yanlixia@buaa.edu.cn} \\
}

\begin{document}
\maketitle

\begin{abstract}
Cooperative of multiple nonholonomic vehicles can be converted into tracking problems of a single-vehicle. The reference trajectory design within distributed features for each vehicle in the group is addressed in this note. The motivation is that nonholonomic vehicles cannot achieve asymptotical stabilization of non-feasible reference signals, and modifications about the virtual reference trajectory design are needed. Reduced-order design and time-varying technique, and some simple geometry tricks are applied to derive the dynamic reference trajectory.
\end{abstract}

\section{Preliminaries}
\subsection{Notations and definitions}
Throughout this paper, $\mathbb{R}$ denotes the set of real numbers, $\|\cdot\|$ represents the Euclidean norm, $|\cdot|$ is the absolute value of a scalar, diag$\{\cdot\}$ denotes the diagonal matrix formed by a vector, $\textit{\textbf{I}}_n$ is an $n$-dimensional identity matrix, $1_n$ is an $n$-dimensional identity vector, $0_n$ denotes an $n$-dimensional zero vector. For a given square matrix, $\lambda(\cdot)$, $\lambda_{m}(\cdot)$ and $\lambda_{M}(\cdot)$ represent the eigenvalue, the smallest and largest eigenvalue, respectively.

\subsection{Graph Theory}
Using graph theory to model the interactions among followers and leaders \citep{RN226} , let $\mathcal{G}=({\mathcal{N}},\mathcal{E},\mathcal{A})$ be the undirected graph with node set $\mathcal{N}=\mathcal{F}\bigcup\mathcal{R}=\{1,...,n,n+1,...,n+m\}$, edge set $\mathcal{E}\subseteq \mathcal{N} \times \mathcal{N}$ and adjacency matrix $\mathcal{A}=\{a_{ij}\}\in \mathbb{R}^{n\times n}$. A direct edge $\{(j,i):i\neq j\}\in \mathcal{E}$ means that node $i$ has access to node $j$. The entry $a_{ij}$ is constant weight and defined as:  $a_{ij}=1$, if $(j,i)\in \mathcal{E}$; $a_{ij}=0$, otherwise. Self interaction is not allowed, i.e., $a_{ii}=0$. The graph $\mathcal{G}$ is called undirected if matrix $\mathcal{A}$ has symmetric weights, i.e., $a_{ij}=a_{ji}, \forall i,j \in \mathcal{N}$. A path from node $i$ to node $j$ denotes an edge sequence $\{(i,j_{1}),(j_1,j_2)...,(j_*,j)\}$, where $i,j_1,j_2,...,j_*,j\in \mathcal{N}$. The Laplacian matrix $\mathcal{L}=\{\mathfrak{{l}}_{ij}\}\in\mathbb{R}^{(n+m)\times (n+m)}$ is denoted as:$\mathfrak{l}_{ij}= - {a_{ij}},~ \text{if}~~i\neq j;\mathfrak{l}_{ij}=\sum\limits_{j = 1}^{n+m} {a_{ij}},~\text{if}~~i = j.$

Partition Laplacian matrix $\mathcal{L}$ into
\begin{equation}\label{eqls1}
\mathcal{L}{\rm{ = }}\left[ {\begin{array}{*{20}{c}}
{{\mathcal{L}_{\rm{1}}}}&{{\mathcal{L}_{\rm{2}}}}\\
{\mathbf{{0}}_{m\times n}}&{\mathbf{{0}}_{m\times m}}
\end{array}} \right]
\end{equation}
where $\mathcal{L}_1\in\mathbb{R}^{n\times n}$ and $\mathcal{L}_2\in\mathbb{R}^{n\times m}$. The interaction graph  $\mathcal{G}$ is supposed to satisfy,
\begin{assumption} \label{A3}
The communication among follower hovercrafts are fixed, undirected and connected. For each individual follower, there exists one path from it to one leader at least.
\end{assumption}
\begin{lemma}\citep{RN1095}
Under Assumption \ref{A3}, the matrix $\mathcal{L}_1$ is positive definite, every entry of $-\mathcal{L}_1^{-1}\mathcal{L}_2$ is non-negative, and the sum of entries of each row of $-\mathcal{L}_1^{-1}\mathcal{L}_2$ equals to 1.
\end{lemma}

\section{Problem Formulation}
Consider $n$ nonholonomic planar vehicles, being viewed as followers with regards to the leaders introduced later. Let $\mathcal{F}=\{1,...n\}$ and define the pose as 
\begin{equation}\label{eq1}
\eta_{F,i}=[p_{F,i}^T,\theta_{F,i}]^T,\forall i \in\mathcal{F},
\end{equation}
where $p_{F,i}=[x_{F,i},y_{F,i}]^T$  and $\theta_{F,i}$ denote position and orientation in Cartesian coordinate system. The planar vehicles are supposed to satisfy a certain nonloholomic constraints, obstructing the Brockett's necessary condition \citep{RN52} for the existence of smooth time-invariant stabilizer. Typical nonholonomic planar vehicles include wheeled mobile robots \citep{lixia}, underactuated surface vehicles\citep{RN402} and underactuated hovercrafts \citep{yan2021robust}.

Besides $n$ followers, consider $m$ virtual leaders, labeled as $n+1$ to $n+m$. At present, we assume that leaders have already achieved a fixed formation pattern and their motions are independent of followers. Define the following leaders' poses,
\begin{equation}\label{ls11}
{\eta _{L,j}}\left( t \right)={\eta _c}\left( t \right) + {d_{L,j}}\in\mathbb{R}^3, j\in\mathcal{R}=\{n+1,...,n+m\},
\end{equation}
where ${\eta _{L,j}}\left( t \right) =[p_{L,j}^T(t),\theta_{L,j}(t)]^T , {\eta _c}\left( t \right) = {\left[ {p_c^T\left( t \right),{\theta _c}\left( t \right)} \right]^T} $ and ${d_{L,j}} = {\left[ {d_{L,jp}^T,{d_{L,j\theta }}} \right]^T}\in \mathbb{R}^3$ is a constant vector with $d_{L,jp}=[d_{L,jx},d_{jy}]^T$. Intuitively, $p_{L,j}=[x_{L,j},y_{L,j}]^T$ and $\theta_{L,j}$ denote the position and orientation of the leader $j$, respectively.
\begin{assumption}\label{a2}
The $\eta_{c}(t)$ is third-differentiable with uniformly bounded time derivatives.
\end{assumption}
\begin{assumption}\label{A4}
The $m$ constant coordinates $d_{L,jp},\forall j\in\mathcal{R}$ span a two-dimensional convex $m-$polygon hull that encloses the origin; and $0\in[\min\{d_{L,j\theta}\},\max\{d_{L,j\theta}\}],\forall j\in\mathcal{R}$.
\end{assumption}
Define two convex sets formed by leader coordinates \eqref{ls11} as
\begin{equation}\label{eq212}
\begin{split}
{\mathcal{C}_{Lp}} &= \left\{ {\left. {\left[ {x,y} \right] \in {\mathbb{R}^2}} \right|\left[ {x,y} \right]^T = \sum\limits_{j = n + 1}^{n + m} {{b_{p,j}}p_{L,j}} } \right\}, \mathcal{C}_{L\theta}= \left\{ {\left. {\theta \in {\mathbb{R}}} \right|\theta = \sum\limits_{j = n + 1}^{n + m} {{b_{p,j}}\theta_{L,j}}} \right\},
\end{split}
\end{equation}
where $b_{p,j}$ satisfy $b_{p,j}\in\mathbb{R}_{[0,1]}$ and $\sum\limits_{j = n + 1}^{n + m} {{b_{p,j}} = 1}$.

Define
\begin{equation}\label{eq110}
\eta_{F,ir}=[p_{F,ir}^T,\theta_{F,ir}]^T \in\mathbb{R}^3, \forall i \in \mathcal{F}.
\end{equation}
where $p_{F,ir}=[x_{F,ir},y_{F,ir}]^T$ denotes position and $\theta_{F,ir}$ is orientation.

Then, the problem is formally stated as: Under Assumptions 1-3, design a virtual reference signal for the holonomic vehicles and derive the conditions within which the nonholonomic vehicles can converge into the convex hull spanned by leaders.

\section{Reference Design and Stability Analysis}
Let us design
\begin{equation}\label{eq12}
\left\{ \begin{split}
{{\dot \eta }_{F,ir}} &= {\varphi _{F,ir}},\\
{{\dot \varphi }_{F,ir}} &= {\rho _{F,ir}},\\
{{\dot \rho }_{F,ir}} &=  - {g_1}{\varphi _{F,ir}} - {g_2}{\rho _{F,ir}} - {g_3}{s_{F,ir}} -  {g_4}\frac{{{s_{F,ir}}}}{{\sqrt {{{\left\| {{s_{F,ir}}} \right\|}^2} + {\gamma _1}^2{e^{ - 2{\gamma _2}t}}} }},
\end{split} \right.
\end{equation}
where $\eta_{F,ir}\triangleq[p_{F,ir}^T,\theta_{F,ir}]^T,g_1>0,g_2>0,g_3>0,g_4>0,\gamma_1>0,\gamma_2>0$ and
\begin{equation}\label{eq13}
\begin{split}
{s_{F,ir}} &= \sum\limits_{j = 1}^n {{a_{ij}}\left( {{\rho _{F,ir}} - {\rho _{F,jr}}} \right)}  + {g_1}\sum\limits_{j = 1}^n {{a_{ij}}\left( {{\varphi _{F,ir}} - {\varphi _{F,jr}}} \right)} \\
&~+ {g_2}\sum\limits_{j = 1}^n {{a_{ij}}\left( {{\eta _{F,ir}} - {\eta _{F,jr}}} \right)} + \sum\limits_{j = n + 1}^{n + m} {{a_{ij}}\left( {{\rho _{F,ir}} - {{\ddot \eta }_{L,j}}} \right)}  \\
 &~+ {g_1}\sum\limits_{j = n+1}^{n + m} {{a_{ij}}\left( {{\varphi _{F,ir}} - {{\dot \eta }_{L,j}}} \right)}  + {g_2}\sum\limits_{j = n+1}^{n + m} {{a_{ij}}\left( {\eta _{F,ir}} - \tilde{\eta}_{L,j} \right)},
\end{split}
\end{equation}
where $\tilde{{\eta}}_{L,j}=\eta_{c}(t)+\mu d_{L,j}$ and $0<\mu<1$. 
\begin{theorem}\label{VRT}
Given Assumptions \ref{a2}-\ref{A3} and $g_1>0, g_2>0, g_3>0, g_4\geq n\bar{\eta}_{c}, \bar{\eta}_{c}= \mathop {\sup }\limits_{t \ge 0} \left\| {\dddot{\eta}_{c} + {g_1}{\ddot \eta}_{c}  + {g_2}{\dot \eta}_{c} } \right\|,\gamma_1>0$ and $\gamma_2>0
$, then the $\eta_{F,ir}=[p_{F,ir}^T,\theta_{F,ir}]^T, \forall i\in \mathcal{F}$ generated by \eqref{eq12} satisfies
\begin{equation}\label{eq112}
\mathop {\lim }\limits_{t \to +\infty } {\eta _{F,r}} =  - \left( {\mathcal{L}_1^{ - 1}{\mathcal{L}_2} \otimes {I_3}} \right)\tilde{\eta}_L,
\end{equation}
where
\begin{equation*}\label{e21x}
\begin{split}
\eta_{F,r}&\triangleq [\eta_{F,1r}^T,...,\eta_{F,nr}^T]^T\in\mathbb{R}^{3n},
\tilde{\eta}_L\triangleq [\tilde{\eta}_{n+1}^T,...,\tilde{\eta}_{n+m}^T]^T\in\mathbb{R}^{3m}.
\end{split}
\end{equation*} 
\end{theorem}
\begin{proof}
  Define three stacked vectors belonging to $\mathbb{R}^{3n}$,
\begin{equation}\label{eq15}
\begin{split}
\varphi_{F,r}&= [\varphi_{F,1r}^T,...,\varphi_{F,nr}^T]^T,\rho_{F,r}= [\rho_{F,1r}^T,...,\rho_{F,nr}^T]^T,
s_{F,r}=[s_{F,1r}^T,...,s_{F,nr}^T]^T.
\end{split}
\end{equation}
By \eqref{eqls1} and \eqref{eq13}, we obtain
\begin{equation}\label{eq16}
\left\{ \begin{split}
{{\dot \eta }_{F,r}} &= {\varphi _{F,r}},\\
{{\dot \varphi }_{F,r}} &= {\rho _{F,r}},\\
{{\dot \rho }_{F,r}} &=  - {g_1}{\varphi _{F,r}} - {g_2}{\rho _{F,r}} - {g_3}{s_{F,r}} - {g_4}P_F{s_{F,r}},
\end{split} \right.
\end{equation}
where $P_F=\mathrm{diag}\{[\frac{\mathbf{1}_3^T}{\sqrt{\|s_{F,1d}\|^2+\gamma_1^2e^{-2\gamma_2t}}},...,\frac{\mathbf{1}_3^T}{\sqrt{\|s_{F,nd}\|^2+\gamma_1^2e^{-2\gamma_2t}}}]\}\in\mathbb{R}^{3n\times 3n}$. Then, arrange ${s}_{F,r}$ as follow,
\begin{equation}\label{eq17}
\begin{split}
{s_{F,r}} &= \left( {{\mathcal{L}_1} \otimes {I_3}} \right)\left( {{\rho _{F,r}} + {g_1}{\varphi _{F,r}} + {g_2}{\eta _{F,r}}} \right) + \left( {{\mathcal{L}_2} \otimes {I_3}} \right)\left[ {{{\ddot {\tilde{\eta}} }_L} + {g_1}{{\dot {\tilde{\eta}} }_L} + {g_2}\tilde{\eta}_L} \right],
\end{split}
\end{equation}
where the facts $\ddot{\tilde{\eta}}_{L}=\ddot{\eta}_L$ and $\dot{\tilde{\eta}}_L=\dot{\eta}_L$ are used. Define
\begin{equation}\label{eq18}
\xi_F=\eta_{F,r}+(\mathcal{L}_1^{-1}\mathcal{L}_2 \otimes I_3)\tilde{\eta}_L,
\end{equation}
and
\begin{equation}\label{eq19}
\begin{split}
\mathbf{S}_F &= {{\ddot \xi }_F} + {g_1}{{\dot
\xi }_F} + {g_2}{\xi _F}\\
 &= {\rho _{F,r}} + \left( {\mathcal{L}_1^{ - 1}{\mathcal{L}_2} \otimes {I_3}} \right){{\ddot {\tilde\eta} }_L} + {g_1}\left[ {{\varphi _{F,r}} + \left( {\mathcal{L}_1^{ - 1}{\mathcal{L}_2} \otimes {I_3}} \right){{\dot {\tilde\eta} }_L}} \right]+ {g_2}\left[ {{\eta _{F,r}} + \left( {\mathcal{L}_1^{ - 1}{\mathcal{L}_2} \otimes {I_3}} \right)\tilde{\eta}_L} \right].
\end{split}
\end{equation}
By \eqref{eq17} and \eqref{eq18}, we get
\begin{equation}\label{eq21}
s_{F,r}=(\mathcal{L}_1\otimes I_3)\mathbf{S}_F.
\end{equation}
It is obvious that the convergence to zero of $\mathbf{S}_F$ can lead to the convergence to zero of $\xi_F$, and hence, achieve $\mathop {\lim }\limits_{t \to \infty }\eta_{F,r}=-(\mathcal{L}_1^{-1}\mathcal{L}_2\otimes I_3)\tilde{\eta}_L$. Differentiate $\mathbf{S}_F$ with respect to time $t$ and yield
\begin{equation}\label{eq20}
\begin{split}
{{\dot{ \mathbf{S}}}_F} &= {{\dddot{\eta}}_{F,r}} + \left( {\mathcal{L}_1^{ - 1}{\mathcal{L}_2} \otimes {I_3}} \right){{\dddot{\tilde\eta}}_L} + {g_1}\left[ {{{\ddot \eta }_{F,r}} + \left( {\mathcal{L}_1^{ - 1}{\mathcal{L}_2} \otimes {I_3}} \right){{\ddot {\tilde\eta} }_L}} \right] + {g_2}\left[ {{{\dot \eta }_{F,r}} + \left( {\mathcal{L}_1^{ - 1}{\mathcal{L}_2} \otimes {I_3}} \right){{\dot {\tilde\eta} }_L}} \right]\\
 &=  - {g_1}{\varphi _{F,r}} - {g_2}{\rho _{F,r}} - {g_3}{s_{F,r}} - {g_4}{P_F}{s_{F,r}} + \left( {\mathcal{L}_1^{ - 1}{\mathcal{L}_2} \otimes {I_3}} \right){{\dddot{{\tilde\eta}}}_L}\\
 &~+ {g_1}\left[ {{\varphi _{F,r}} + \left( {\mathcal{L}_1^{ - 1}{\mathcal{L}_2} \otimes {I_3}} \right){{\ddot {\tilde\eta} }_L}} \right]+ {g_2}\left[ {{\rho _{F,r}} + \left( {\mathcal{L}_1^{ - 1}{\mathcal{L}_2} \otimes {I_3}} \right){{\dot {\tilde\eta} }_L}} \right]\\
 &=  - {g_3}{s_{F,r}} - {g_4}{P_F}{s_{F,r}} + \left( {\mathcal{L}_1^{ - 1}{\mathcal{L}_2} \otimes {I_3}} \right)\left( {{{\dddot{\tilde\eta}}_L} + {g_1}{{\ddot {\tilde\eta} }_L} + {g_2}{{\dot {\tilde\eta} }_L}} \right)\\
 &=- {g_3}(\mathcal{L}_1\otimes I_3){S_{F}} - {g_4}{P_F}(\mathcal{L}_1\otimes I_3){S_{F}}+ \left( {\mathcal{L}_1^{ - 1}{\mathcal{L}_2} \otimes {I_3}} \right)\left( {{{\dddot{{\tilde\eta}}}_L} + {g_1}{{\ddot {\tilde\eta} }_L} + {g_2}{{\dot {\tilde\eta} }_L}} \right).
\end{split}
\end{equation}
Using the property that $\mathcal{L}_1$ is positive definite, choose a Lyapunov candidate as,
\begin{equation}\label{eq22}
V_1=0.5\mathbf{S}_F^T(\mathcal{L}_1\otimes I_3)\mathbf{S}_F,
\end{equation}
which satisfies $0.5\lambda_{\min} (\mathcal{L}_1)\|\mathbf{S}_F\|^2\leq V_1 \leq0.5\lambda_{\max} (\mathcal{L}_1)\|\mathbf{S}_F\|^2$ and $\|\mathbf{S}_F\|^2\geq 2 \frac{V_1}{\lambda_{\max}(\mathcal{L}_1)}$. The time derivative of \eqref{eq22} on the trajectory of \eqref{eq20} is
\begin{equation}\label{eq23}
\begin{split}
{\dot V}_1 &=  - {g_3}\mathbf{S}_F^T{\left( {{\mathcal{L}_1} \otimes {I_3}} \right)^2}{\mathbf{S}_F} - {g_4}{P_F}\mathbf{S}_F^T{\left( {{\mathcal{L}_1} \otimes {I_3}} \right)^2}{\mathbf{S}_F}+ \mathbf{S}_F^T\left( {{\mathcal{L}_1} \otimes {I_3}} \right)\left( {\mathcal{L}_1^{ - 1}{\mathcal{L}_2} \otimes {I_3}} \right)\left( {{{\dddot {\tilde\eta} }_L} + {g_1}{{\ddot {\tilde\eta} }_L} + {g_2}{\dot{{\tilde\eta}} _L}} \right)\\
 &=  - {g_3}\mathbf{S}_F^T{\left( {{\mathcal{L}_1} \otimes {I_3}} \right)^2}{\mathbf{S}_F} - {g_4}{P_4}s_{F,r}^T{s_{F,r}} + {s_{F,r}^T}\left( {\mathcal{L}_1^{ - 1}{\mathcal{L}_2} \otimes {I_3}} \right)\left( {{{\dddot {\tilde\eta} }_L} + {g_1}{{\ddot {\tilde\eta} }_L} + {g_2}{\dot{{\tilde\eta}} _L}} \right),
\end{split}
\end{equation}
where the equation $s_{F,r}^T=\mathbf{S}_F^T(\mathcal{L}_1\otimes I_3)$ is utilized. It then by $-\mathcal{L}_1^{-1}\mathcal{L}_2 1_m=1_n$ follows that
\begin{equation}\label{sxxsa}
\begin{split}
s_{F,r}^T\left( {\mathcal{L}_1^{ - 1}{\mathcal{L}_2} \otimes {I_3}} \right)\left( {{{\dddot{\tilde\eta}}_L} + {g_1}{{\ddot {\tilde\eta} }_L} + {g_2}{{\dot {\tilde\eta} }_L}} \right)
&~=s_{F,r}^T\left( {\mathcal{L}_1^{ - 1}{\mathcal{L}_2} \otimes {I_3}} \right)\left[1_m\otimes(\dddot{\eta}_c+g_1\ddot{\eta}_c+g_2\dot{\eta}_c)\right]\\
&~=s_{F,r}^T\left[ {\mathcal{L}_1^{ - 1}{\mathcal{L}_2}1_m \otimes (\dddot{\eta}_c+g_1\ddot{\eta}_c+g_2\dot{\eta}_c)} \right]\\
&~=-s_{F,r}^T\left[ {1_n \otimes (\dddot{\eta}_c+g_1\ddot{\eta}_c+g_2\dot{\eta}_c)} \right]\\
&~\le n\sum\limits_{i = 1}^n {\left\| {s_{F,ir}^T} \right\|} {{\bar \eta }_{c}},
\end{split}
\end{equation}
where the fact $\tilde{\eta}_{L,j}^{(q)}=\eta_{c}^{(q)}(t), \forall j\in\mathcal{R}, q\in \mathbb{Z}_{\geq1}$ is also applied.
The combination of \eqref{eq23} and \eqref{sxxsa} leads to
\begin{small}
\begin{equation}\label{eq24}
\begin{split}
{\dot V}_1 &\le  - {g_3}\mathbf{S}_F^T{\left( {{\mathcal{L}_1} \otimes {I_3}} \right)^2}{\mathbf{S}_F} - \sum\limits_{i = 1}^n {\frac{{{g_4}s_{F,ir}^T{s_{F,ir}}}}{{\sqrt {{{\left\| {{s_{F,ir}}} \right\|}^2} + \gamma _1^2{e^{ - 2{\gamma _2}t}}} }} + } \sum\limits_{i = 1}^n {\left\| {s_{F,ir}^T} \right\|n\bar{\eta}_{c}} \\
 &=  - {g_3}\mathbf{S}_F^T{\left( {{\mathcal{L}_1} \otimes {I_3}} \right)^2}{\mathbf{S}_F} - \sum\limits_{i = 1}^n {\left( {\frac{{{g_4}{{\left\| {{s_{F,ir}}} \right\|}^2}}}{{\sqrt {{{\left\| {{s_{F,ir}}} \right\|}^2} + \gamma _1^2{e^{ - 2{\gamma _2}t}}} }} - \left\| {s_{F,ir}^T} \right\|n{{\bar \eta }_{c}}} \right)}.
\end{split}
\end{equation}
\end{small}
After some direct computations,
\begin{equation}\label{eq26}
\begin{split}
{\dot V}_1 &\le  - {g_3}\lambda _{\min }^2\left( {{\mathcal{L}_1}} \right){\left\| {{S_{F,r}}} \right\|^2} \\
&~- \sum\limits_{i = 1}^n {\left( {\frac{{\left( {{g_4} - n{{\bar \eta }_{c}}} \right){{\left\| {{s_{F,ir}}} \right\|}^2}}}{{\left\| {{s_{F,ir}}} \right\| + {\gamma _1}{e^{ - {\gamma _2}t}}}} - \frac{{\left\| {{s_{F,ir}}} \right\|}}{{\left\| {{s_{F,ir}}} \right\| + {\gamma _1}{e^{ - {\gamma _2}t}}}}n{{\bar \eta }_{c}}{\gamma _1}{e^{ - {\gamma _2}t}}} \right)} \\
 &\le  - \frac{{2{g_3}\lambda _{\min }^2\left( {{\mathcal{L}_1}} \right)}}{{{\lambda _{\max }}\left( {{\mathcal{L}_1}} \right)}}V_1 + \sum\limits_{i = 1}^n {\frac{{\left\| {{s_{F,ir}}} \right\|}}{{\left\| {{s_{F,ir}}} \right\| + {\gamma _1}{e^{ - {\gamma _2}t}}}}n{{\bar \eta }_{c}}{\gamma _1}{e^{ - {\gamma _2}t}}}
\end{split}
\end{equation}
where the $g_4\geq n\bar{\eta}_{L,*}$ is used. Let $\lambda_1=\frac{2g_3\lambda_{\min}^2(\mathcal{L}_1)}{\lambda_{\max}(\mathcal{L}_1)}$, convert the second inequality of \eqref{eq26} into
\begin{equation}\label{eq27}
{\dot V}_1\leq -\lambda_1{ V}_1+n^2\gamma_1\bar{\eta}_{c}e^{-\gamma_2t}.
\end{equation}
Integrating both sides of \eqref{eq27} by comparison principle \citep{NLS}, we yield
\begin{equation}\label{eq28}
\begin{split}
V_1\left( t \right) &\le {e^{ - {\lambda _1}t}}V_1\left( 0 \right) + \int_0^t {{e^{ - \left( {{\lambda _1} - \chi } \right)}}} n^2{\gamma _1}{{\bar \eta }_{c}}{e^{ - {\lambda _1}\chi }}{\rm{d}}\chi \\
 &= \left\{ \begin{array}{l}
{e^{ - {\lambda _1}t}}V_1\left( 0 \right) + n^2{\gamma _1}{{\bar \eta }_{c}}t{e^{ - {\lambda _1}t}},\;\;\;\;{\rm{if}}\;{\lambda _1} = {\gamma _2},\\
{e^{ - {\lambda _1}t}}V_1\left( 0 \right) + n^2{\gamma _1}{{\bar \eta }_{c}}\displaystyle\frac{{{e^{ - {\lambda _1}t}} - {e^{ - {\gamma _2}t}}}}{{{\gamma _2} - {\lambda _1}}},\;\;\;\;{\rm{if}}\;{\lambda _1} \ne {\gamma _2},
\end{array} \right.
\end{split}
\end{equation}
and find out that $V_1$ converges to zero globally exponentially. Hence, $\mathbf{S}_F$ is globally exponentially convergent. By \eqref{eq19}, we obtain $\ddot{\xi}_F=-g_1\dot{\xi}_F-g_2\xi_F+\mathbf{S}_F$, which further implies that $\xi_F$ and $\dot{\xi}_F$ are globally exponentially convergent to zero. Henceforth, we have
\begin{equation}\label{eqx36}
\mathop {\lim }\limits_{t \to +\infty } {\eta _{F,r}} =  - \left( {\mathcal{L}_1^{ - 1}{\mathcal{L}_2} \otimes {I_3}} \right)\tilde{\eta}_L.
\end{equation}
Because the sum of each row of $-\mathcal{L}_1^{-1}\mathcal{L}_2$ equals $1$, we obtain $p_{F,ir}\to\tilde{\mathcal{C}}_{Lp}$ and $\theta_{F,ir}\to\tilde{\mathcal{C}}_{L\theta}$ as $t\to +\infty$. Additionally, by the above analysis and Assumption 2.2, it is direct to verify that the derivatives $\dot{\eta}_{F,ir}=[\dot{p}_{F,ir}^T,\dot{\theta}_{F,ir}]^T, \ddot{\eta}_{F,ir}=[\ddot{p}_{F,ir}^T,\ddot{\theta}_{F,ir}]^T$ and $\dddot{\eta}_{F,ir}=[\dddot{p}_{F,ir}^T,\dddot{\theta}_{F,ir}]^T$ are bounded for all $t\geq 0$. This completes the proof.
\end{proof}
According to Lemma 1, the Theorem 1 tells that the virtual reference signal propagated by \eqref{eq12} would converge into the convex set formed by vectors $\tilde{\eta}_{L,j}=\eta_c(t)+\mu d_{L,j},\forall j\in\mathcal{R}$. Note that the constant $\mu$ scales down that original convex hulls. This scaling is an essential step for steering the hovercrafts converge into the original convex hulls spanned by leaders. Let $\tilde{\mathcal C}_{Lp}$ and $\tilde{\mathcal C}_{L\theta}$ denote the convex hulls formed by $\tilde{\eta}_{L,j},\forall j\in\mathcal{R}$, we have $\tilde{\mathcal{C}}_{Lp}\subseteq\mathcal{C}_{Lp}$ and $\tilde{\mathcal{C}}_{L\theta}\subseteq\mathcal{C}_{L\theta}$.
Define $\bar{\mathcal{C}}_{Lp}=\mathbb{R}^2-{\mathcal{C}}_{Lp}$ and $\bar{\mathcal{C}}_{L\theta}=\mathbb{R}-\mathcal{C}_{L\theta}$ and propose the following lemma,
\begin{theorem}
Given points $p_1\in\tilde{\mathcal{C}}_{Lp},p_2\in\bar{\mathcal{C}}_{Lp}$ and $\theta_1\in\tilde{\mathcal{C}}_{L\theta},\theta_2\in\bar{\mathcal{C}}_{L\theta}$, there exist constants $\alpha_p$ and $\alpha_\theta$ such that
\begin{equation}\label{eqnelx}
\begin{split}
\|p_1-p_2\|&>\alpha_p\triangleq(1-\mu)\min\left\{\frac{|d_{L,jx}d_{L,iy}-d_{L,ix}d_{L,jy}|}{\sqrt{(d_{L,jy}-d_{L,iy})^2+(d_{L,ix}-d_{L,iy})^2}}\right\},\\
|\theta_1-\theta_2|&>\alpha_\theta\triangleq(1-\mu)\min\left\{|d_{L,j\theta}|\right\},\forall i,j \in\mathcal{R}.
\end{split}
\end{equation}
\end{theorem}
\begin{proof}
Suppose that the $\mathcal{C}_{Lp}$ is composed of $m$ edges, labeled as $E_{L,n+1},...,E_{L,n+m}$ by connecting the coordinates $p_{L,i},i\in\mathcal{R}$ in couterclockwise(or clockwise) circular order. Different from the edge for communication in graph theory, the edge $E_{L,i}, i \in \mathcal{R}$ here refers to the line segment connected by position coordinates $p_{L,i}=[x_{L,i},y_{L,i}]^T$ and $p_{L,j}=[x_{L,j},y_{L,j}]^T$, where, without making any confusion, the subscript is chosen as: $j=i+1$ if $i\neq n+m$; otherwise $j=n+1, \forall i\in\mathcal{R}$. By definition \eqref{ls11}, the vertex coordinates of the edge $E_{L,i}$ are
\begin{equation}\label{eql2x}
\begin{split}
p_{L,i}&=[x_c(t)+d_{L,ix},y_c(t)+d_{L,iy}]^T,\\
p_{L,j}&=[x_c(t)+d_{L,jx},y_c(t)+d_{L,jy}]^T.
\end{split}
\end{equation}
Then, the scaled coordinates of \eqref{eql2x}, according to $\tilde{\eta}_{L,j}$ in\eqref{eq13}, become
\begin{equation}\label{eql2y}
\begin{split}
\tilde{p}_{L,i}&=[x_c(t)+\mu d_{L,ix},y_c(t)+ \mu d_{L,iy}]^T,\\
\tilde{p}_{L,j}&=[x_c(t)+\mu d_{L,jx},y_c(t)+\mu d_{L,jy}]^T.
\end{split}
\end{equation}
We denote the edge connected by $\tilde{p}_{L,i}$ and $\tilde{p}_{L,j}$ as $\tilde{E}_{L,i}$. The relationship between the paired edges $(E_{L,i},\tilde{E}_{L,i})$, by \eqref{eql2x} and \eqref{eql2y}, satisfies
\begin{equation}\label{ex1x}
\left\| {{p_{L,i}} - {p_{L,j}}} \right\| = \mu \left\| {{{\tilde p}_{L,i}} - {{\tilde p}_{L,j}}} \right\|,
\end{equation}
which means that the length of edge $E_{L,i}$  is proportional to that of its scaled part $\tilde{E}_{L,i}$ at a ratio of $\mu$. This is true for all $m$ paired edges $(E_{L,i},\tilde{E}_{L,i}),\forall i \in \mathcal{R}$. Therefore, the coordinates $\tilde{p}_{L,i},\forall i \in \mathcal{R}$ span a sub-hull $\tilde{\mathcal{C}}_{Lp}$ that is not only similar to $\mathcal{C}_{Lp}$ but also is convex. According to Assumption \ref{A4} and the fact that $0<\mu<1$, any point in $\tilde{\mathcal{C}}_{Lp}$ is also in $\mathcal{C}_{Lp}$, namely, $\tilde{\mathcal{C}}_{Lp}\subseteq\mathcal{C}_{Lp}$. For $\tilde{\mathcal{C}}_{L\theta}$, we have $\tilde{\mathcal{C}}_{L\theta}=\{\theta\in\mathbb{R}|\theta\in[\theta_c(t)+\mu\min\{d_{L,j\theta}\},\theta_c(t)+\mu\max\{{d}_{L,j\theta}\}],\forall j \in \mathcal{R}\}$. Thus, it is direct to prove that $\tilde{\mathcal{C}}_{L\theta}\subseteq \mathcal{C}_{L\theta}$.

To prove the property 2, we calculate the straight-line equations of paired edges $E_{L,i}$ and $\tilde{E}_{L,i}$ as follows,
\begin{small}
\begin{subequations}\label{1s1}
\begin{equation}\label{line1x}
\begin{split}
&(d_{L,jy}-d_{L,iy})x+(d_{L,ix}-d_{L,iy})y+[x_c(t)+d_{L,jx}][y_c(t)+d_{L,iy}]-[x_c(t)+d_{L,ix}][y_c(t)+d_{L,jy}]=0;
\end{split}
\end{equation}
\begin{equation}\label{line1x2}
\begin{split}
 & \mu(d_{L,jy}-d_{L,iy})x+\mu(d_{L,ix}-d_{L,iy})y+[x_c(t)+\mu d_{L,jx}][y_c(t)+\mu d_{L,iy}]-[x_c(t)+\mu d_{L,ix}][y_c(t)+\mu d_{L,jy}]=0.
\end{split}
\end{equation}
\end{subequations}
\end{small}
Obviously, the straight lines described by \eqref{line1x} and \eqref{line1x2} are parallel; and so are the paired edges $E_{L,i}$ and $\tilde{E}_{L,i}$. According to basic geometric knowledge, the distance between $E_{L,i}$ and $\tilde{E}_{L,i}$ can be calculated as,
\begin{equation}\label{ewx1}
\mathrm{dis}(E_{L,i},\tilde{E}_{L,i})=(1-\mu)\frac{|d_{L,jx}d_{L,iy}-d_{L,ix}d_{L,jy}|}{\sqrt{(d_{L,jy}-d_{L,iy})^2+(d_{L,ix}-d_{L,iy})^2}},
\end{equation}
where `dis' means `the distance of'. Note that the Assumption \ref{A4} ensures that both the numerator and denominator of \eqref{ewx1} are greater than zero. Therefore, $\mathrm{dis}(E_{L,i},\tilde{E}_{L,i})>0$ holds. Let us choose
\begin{equation}\label{aplxpp}
\alpha_p=(1-\mu)\min\left\{\frac{|d_{L,jx}d_{L,iy}-d_{L,ix}d_{L,jy}|}{\sqrt{(d_{L,jy}-d_{L,iy})^2+(d_{L,ix}-d_{L,iy})^2}}\right\},\forall i,j \in\mathcal{R}.
\end{equation}
The constant $\alpha_p$ is selected as the minimum distance between the parallel paired edges $E_{L,i}$ and $\tilde{E}_{L,i},\forall i \in\mathcal{R}$. Hence, the distance between any point in $\tilde{\mathcal{C}}_{Lp}$ and the other one outside $\mathcal{C}_{Lp}$ is greater than $\alpha_p$. Following the same routine above, we can calculate $\alpha_\theta$ as
\begin{equation}\label{apxhthe}
\alpha_\theta=(1-\mu)\min\left\{|d_{L,j\theta}|\right\},\forall j\in\mathcal{R}.
\end{equation}
A graphical version of the above analysis refers to Fig. 2. The choice of $\alpha_p$ and $\alpha_\theta$ is to guarantee that the distance from any point outside the original convex hull to that in the scaled one is greater than the minimum distance between the `edges' before and after scaling. This completes the proof.
\begin{figure}[h]\label{F2}
\centering     
\includegraphics[width=.80\columnwidth]{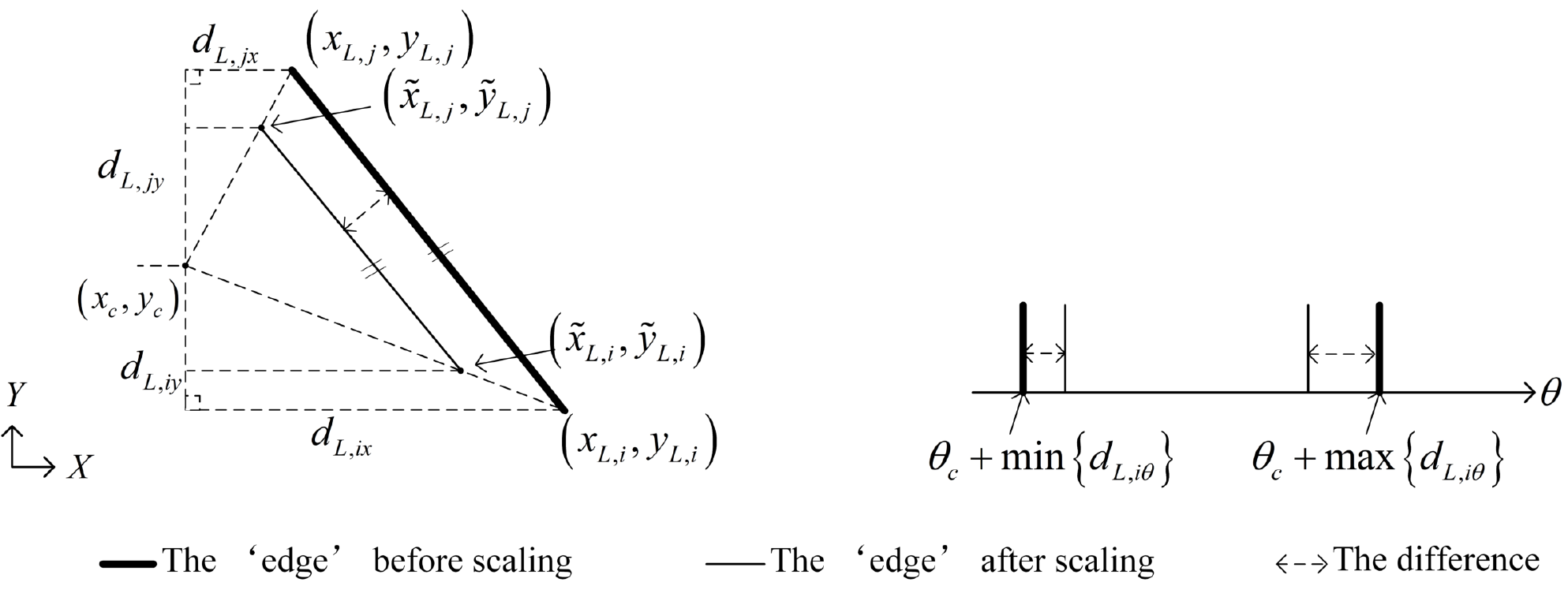}
\caption{Illustration of constructing the convex sub-hulls.}
\end{figure}
\end{proof}

By the expressions of \eqref{aplxpp} and \eqref{apxhthe}, both $\alpha_p$ and $\alpha_\theta$ can be set arbitrarily small by decreasing~$\mu$, so that the scaled convex sub-hulls $\tilde{\mathcal{C}}_{Lp}$ and $\tilde{\mathcal{C}}_{L\theta}$ can approximate the original ones with arbitrarily small difference. Thisimplies that the control design for the vehicle in question ought to ensure that the ultimate bounds of pose tracking errors $p_{F,i}-p_{F,ir}$ and $\theta_{F,i}-\theta_{F,ir}$ be tunable and arbitrarily small.
\bibliographystyle{unsrtnat}
\bibliography{references}

\end{document}